\documentclass{article}
\usepackage{amsmath,amsthm,amssymb}
\usepackage{array, longtable}
\newtheorem{Theorem}{Theorem}[section]
\newtheorem{lem}[Theorem]{Lemma}

\numberwithin{equation}{section}
\numberwithin{table}{section}

\begin{document}

\title{New constructions of  MDS codes with complementary duals}

\insert\footins{\footnotesize {\it Email addresses}:
{bocong\_chen@yahoo.com (B. Chen)},
 hwliu@mail.ccnu.edu.cn (H. Liu). }

\author{Bocong Chen$^{1}$,~~Hongwei Liu$^{2}$}

\date{\small
${}^1$School of Mathematics, South China University of Technology, Guangzhou,
Guangdong, 510641, China\\
${}^2$School of Mathematics and Statistics,
Central China Normal University,
Wuhan, Hubei, 430079, China         }

\maketitle

\begin{abstract}
Linear complementary-dual  (LCD  for short) codes
 are   linear codes that intersect  with their duals  trivially.
 LCD codes have been used  in certain  communication systems.
 It is recently found that
LCD codes  can be applied  in cryptography. This application of LCD codes  renewed
the interest in the construction  of LCD codes having a  large minimum distance.
MDS  codes are optimal in the sense that
the  minimum distance  cannot be improved for given length and code size.
Constructing LCD MDS codes is thus of
significance in theory and practice.
Recently, Jin (\cite{Jin}, IEEE Trans. Inf. Theory,  2016)
constructed several classes of LCD MDS codes through generalized Reed-Solomon
codes. In this paper, a different approach is proposed to obtain  new LCD MDS codes
from generalized Reed-Solomon
codes.
Consequently, new code constructions  are provided  and  certain  previously known results in \cite{Jin} are extended.


\end{abstract}

\section{Introduction}
A {\it linear complementary-dual}  (LCD  for short) code
is a linear code $\mathcal{C}$ that intersects  with its dual $\mathcal{C}^\perp$  trivially,
i.e.,  $\mathcal{C}\bigcap\mathcal{C}^\perp=\{\mathbf{0}\}$.
This class of codes was  introduced  by Massey   \cite{Massey}, where  he showed that
asymptotically good LCD codes exist.
In a follow-up paper \cite{Yang},
Yang and Massey   gave a necessary and sufficient condition for a cyclic code to be an LCD code.
Necessary and sufficient conditions for certain classes of quasi-cyclic codes to be LCD codes  were obtained in \cite{Esmaeili}.
Dinh \cite{Dinh} proved that any $\lambda$-constayclic codes with $\lambda\notin\{1,-1\}$ are LCD codes.
G\"{u}neri {\it et al.} \cite{Guneri} studied LCD quasi-cyclic codes,
including the algebraic characterizations, the asymptotic properties and the code constructions.
Dougherty {\it et al.} \cite{Dougherty} presented a linear programming bound on the largest size of an LCD code
of given length and minimum distance.
The parameters of several classes of LCD BCH codes  were  explicitly determined in \cite{LiC1,LiC2,LiS}.
Hou and Oggier \cite{Hou} established  a corresponding relationship between  lattices and LCD codes.

LCD codes have been used  in certain  communication systems.
Recently,  Carlet and Guilley \cite{Carlet}
found that  LCD codes  can be applied  in cryptography.
This application of LCD codes renewed
the interest in the construction of LCD codes having  a  large minimum distance.
Maximum distance separable (MDS)
codes are optimal in the sense that
no code of length $n$
with $M$ codewords has a larger minimum distance than an MDS code with  length  $n$
and size $M$.
Constructing LCD MDS codes is thus of
significance in theory and practice.
The class of generalized Reed-Solomon (GRS) codes is  probably the best known family of MDS codes.
Therefore, it is natural to construct LCD MDS codes through GRS codes.
Recently,
Jin \cite{Jin}
constructed several classes of LCD MDS codes by using two classes of disjoint GRS codes.
The existence question about LCD MDS codes over a finite field of even characteristic has been completely addressed in
\cite[Theorem IV.2]{Jin}.
Some other constructions of LCD MDS codes are known, e.g., see \cite{Me},  \cite{Sari} and \cite{Zhu}.

The purpose of this paper is to find new LCD MDS codes from GRS codes or extended GRS codes.
By virtue of \cite[Theorem IV.2]{Jin}, we only consider the construction of  LCD MDS codes over a finite field with odd characteristic.
A different approach from \cite{Jin} is proposed to obtain  new LCD MDS codes; specifically, some results of \cite{Jin} are extended.
Note that
the dual of an LCD code is an LCD code again. As a consequence, we always restrict ourself to $k$-dimensional  codes of length $n$ with
$1<k\leq \lfloor n/2\rfloor$, where $\lfloor a\rfloor$ denotes the integer part of a  real number $a$.
The main results of this paper are summarized as follows.
\begin{Theorem}\label{1.1}
Let $q>3$ be an odd prime power.
Let $n>1$ and $k$ be positive integers with $1<k\leq\lfloor n/2\rfloor$.
Then there
exists a $q$-ary $[n,k]$ LCD MDS code whenever one of the following conditions holds.
\begin{itemize}
\item[$(1)$](see Theorem \ref{generalized1})
 $n=q+1$.

\item[$(2)$](see Theorem \ref{divides-theorem})
 $n>1$ is a divisor of $q-1$.

\item[$(3)$](see Theorem \ref{power})
$q=p^e$ and $n=p^\ell$,  where   $p$ is  prime  and $1\leq \ell\leq e$.

\item[$(4)$](see Theorem \ref{3.6})
$n<q$ and $n+k\geq q+1$.

\item[$(5)$](see Theorem \ref{last-theorem})
$n<q$ and $2n-k<q\leq2n$.
\end{itemize}
\end{Theorem}
At this point we make several remarks. The first part of Theorem \ref{1.1} improves
\cite[Theorem IV.4]{Jin} by removing the even $k$ constraint, and
the existence question about LCD MDS codes of length $n=q+1$ is completely addressed.
The second and  the third parts, respectively,    say that
$k$-dimensional LCD MDS codes  of length $q-1$ and $q$ over $\mathbb{F}_q$ exist.
The last two parts of Theorem \ref{1.1} imply  that $[n,k]$ LCD MDS codes    exist if $n$ and $k$ are close to $q$.
The conclusions of our results are different from those in \cite[Theorem IV.3]{Jin} and \cite[Theorem IV.6]{Jin}, where it requires that
the value of $q$ is much bigger than the value of $n$.

The remainder of this paper is organized as follows.
Basic
notations and results about  GRS
codes and extended GRS codes are provided in Section \ref{preliminaries}.
The main results included in Theorem \ref{1.1} are presented in Section \ref{sec-main}.

\section{Preliminaries}\label{preliminaries}
In this section,    we review some basic notations and results
about  generalized Reed-Solomon codes   and extended generalized Reed-Solomon codes.
For the details, the reader is referred to \cite{Ling} or \cite{Mac}.
Let $\mathbb{F}_q$ be the finite field of order $q$,
let $n$ be a positive integer with $1<n\leq q$,
and let
$\mathbf{a}=(\alpha_1, \alpha_2, \cdots, \alpha_n),$
where $\alpha_i$ ($1\leq i\leq n$) are  distinct elements of $\mathbb{F}_q$.
Fix $n$ nonzero elements $v_1, v_2, \cdots, v_n$ of $\mathbb{F}_q$ ($v_i$ are not necessarily distinct).
For $1\leq k\leq n$, the $k$-dimensional {\it generalized Reed-Solomon code} (GRS code for short) of length $n$ associated with
$\mathbf{a}=(\alpha_1, \alpha_2, \cdots, \alpha_n)$
and $\mathbf{v}=(v_1,v_2,\cdots,v_n)$
is defined to be
\begin{equation}\label{GRS}
GRS_{k}\big(\mathbf{a},\mathbf{v}\big)=
\Big\{\big(v_1f(\alpha_1), v_2f(\alpha_2), \cdots, v_nf(\alpha_n)\big)\,\Big{|}\,f(X)\in \mathbb{F}_q[X], \deg f(X)\leq k-1\Big\}.
\end{equation}
Clearly, $GRS_{k}(\mathbf{a},\mathbf{v})$ has a generator matrix
\begin{equation}\label{matrix}
G=\left(
    \begin{array}{ccccccc}
      v_1 & v_2 & \cdots & v_n &     \\
      v_1\alpha_1 & v_2\alpha_2 & \cdots & v_n\alpha_{n}     \\
      v_1\alpha_1^2 & v_2\alpha_2^2 & \cdots & v_n\alpha_{n}^2     \\
      \vdots & \vdots & \ddots & \vdots&     \\
      v_1\alpha_1^{k-1} & v_2\alpha_2^{k-1} & \cdots & v_n\alpha_{n}^{k-1}     \\
    \end{array}
  \right).
\end{equation}
It is well known that the code $GRS_{k}(\mathbf{a},\mathbf{v})$ is a $q$-ary $[n,k,n-k+1]$-MDS code and the dual of  a GRS code is again a GRS code;
indeed,
$$
GRS_{k}(\mathbf{a},\mathbf{v})^\perp=GRS_{n-k}(\mathbf{a},\mathbf{v'})
$$
for some $\mathbf{v'}=(v'_1, v'_2, \cdots, v'_n)$
with $v'_i\neq0$ for all $1\leq i\leq n$ (e.g., see \cite{Ling} or \cite{Mac}).
The vector $\mathbf{v'}$  can be chosen as any vector that generates the
dual of $GRS_{n-1}(\mathbf{a},\mathbf{v})$.

Let $\mathbf{1}$ denote the all-one row vector with appropriate length.
The dual of $GRS_{k}(\mathbf{a},\mathbf{1})$ is $GRS_{n-k}(\mathbf{a},\mathbf{u})$, where
$\mathbf{u}=(u_1,u_2,\cdots, u_n)$ with $u_i=\prod_{1\leq j\leq n,j\neq i}(\alpha_i-\alpha_j)^{-1}$ for $1\leq i\leq n$
(e.g., see \cite[Lemma 2.3]{JinX}). More precisely, one has
\begin{equation}\label{important}
GRS_{k}(\mathbf{a},\mathbf{1})^\perp=GRS_{n-k}(\mathbf{a},\mathbf{u})=
\Big\{\big(u_1g(\alpha_1),   \cdots, u_ng(\alpha_n)\big)\,\Big{|}\,g(X)\in \mathbb{F}_q[X], \deg g(X)\leq n-k-1\Big\}.
\end{equation}

GRS codes  are probably the best known family of MDS codes.
Obviously, GRS codes  exist for any length $n\leq q$ and any dimension $k\leq n$.
GRS codes    of length $n$ can be extended to codes of length $n+1$  while preserving the MDS
property by appending to $G$   an extra column of the form   $(0,0,\cdots, 0, \beta)^T$ with $\beta$ being a nonzero element of $\mathbb{F}_q$.
In this paper, we consider extended GRS codes of length $q+1$.
Label  the elements of  $\mathbb{F}_q=\{\alpha_1,\alpha_2, \cdots, \alpha_q\}$.
Given a vector $\mathbf{v}=(v_1,v_2,\cdots, v_q)$ with $v_j\neq0$ for all $1\leq j\leq q$, the $k$-dimensional
{\it extended GRS code} of length $q+1$ associated with $\mathbf{a}=(a_1,a_2,\cdots,a_q)$
and $\mathbf{v}$ is defined as
\begin{equation}\label{extended-grs}
GRS_{k}\big(\mathbf{a},\mathbf{v},\infty\big)=\Big\{\big(v_1f(\alpha_1),\cdots, v_qf(\alpha_q), f_{k-1}\big)\,\Big{|}\,f(X)\in \mathbb{F}_q[X], \deg f(X)\leq k-1\Big\},
\end{equation}
where $f_{k-1}$ is the coefficient of $X^{k-1}$ of the polynomial $f(X)$.
It is true that $GRS_{k}(\mathbf{a},\mathbf{v},\infty)$ is a $q$-ary  MDS code with parameters $[q+1,k,q-k+2]$
(e.g., see \cite{Ling} or \cite{Mac}).
The code $GRS_{k}(\mathbf{a},\mathbf{v},\infty)$ has a generator matrix
\begin{equation*}
G_{\infty}=\left(
    \begin{array}{ccccc}
      v_1 & v_2 & \cdots   & v_q   & 0 \\
      v_1\alpha_1 & v_2\alpha_2 & \cdots & v_q\alpha_{q}  & 0 \\
      v_1\alpha_1^2 & v_2\alpha_2^2 & \cdots & v_q\alpha_{q}^2  & 0 \\
      \vdots & \vdots & \ddots & \vdots & \vdots \\
      v_1\alpha_1^{k-1} & v_2\alpha_2^{k-1} & \cdots & v_q\alpha_{q}^{k-1} & 1 \\
    \end{array}
  \right).
\end{equation*}
The dual of $GRS_{k}\big(\mathbf{a},\mathbf{1},\infty\big)$
can be determined explicitly; in fact,
it is not hard to verify  that the dual of $GRS_{k}\big(\mathbf{a},\mathbf{1},\infty\big)$ is
$$
GRS_{k}\big(\mathbf{a},\mathbf{1},\infty\big)^\perp=
\Big\{\big(g(\alpha_1), g(\alpha_2), \cdots, g(\alpha_q), g_{q-k}\big)\,\Big{|}\,g(X)\in \mathbb{F}_q[X], \deg g(X)\leq q-k\Big\},
$$
where $g_{q-k}$ is the coefficient of $X^{q-k}$ of the polynomial $g(X)$ (e.g., see \cite[Lemma 7.7]{Ball} or \cite[Lemma 2.3]{JinX}).

\section{Constructions of  LCD MDS codes}\label{sec-main}
The purpose of this section is to find LCD codes among the family of GRS codes or among the family of extended GRS codes.
Consequently  the resulting codes are simultaneously LCD and MDS.
As mentioned in the first section, the existence question about LCD MDS codes over a finite field of even characteristic has been completely addressed in
\cite{Jin}.
We therefore only consider GRS codes over a finite field with odd characteristic.

Note that
the dual of an LCD code is an LCD code again. As a consequence, we always restrict ourself to $k$-dimensional  codes of length $n$ with
$1<k\leq \lfloor n/2\rfloor$, where $\lfloor a\rfloor$ denotes the integer part of a  real number $a$.

We begin with the following lemma which is  useful for the construction of LCD GRS codes.
\begin{lem}\label{lem}
Suppose  $GRS_k(\mathbf{a},\mathbf{v})$ is the   GRS code  associated with $\mathbf{a}$
and $\mathbf{v}$,  as in (\ref{GRS}).
A typical codeword $\mathbf{c}=\big(v_1f(\alpha_1), v_2f(\alpha_2), \cdots, v_nf(\alpha_n)\big)$ of $GRS_k(\mathbf{a},\mathbf{v})$
is contained in
$GRS_k(\mathbf{a},\mathbf{v})^\perp$ if and only if
a polynomial $g(X)\in\mathbb{F}_q[X]$ with  $\deg g(X)\leq n-k-1$ can be found such that
\begin{equation}\label{lemma-equation}
\big(v_1^2f(\alpha_1), v_2^2f(\alpha_2), \cdots, v_n^2f(\alpha_n)\big)=
\big(u_1g(\alpha_1), u_2g(\alpha_2), \cdots, u_ng(\alpha_n)\big),
\end{equation}
where $u_i=\prod_{1\leq j\leq n,j\neq i}(\alpha_i-\alpha_j)^{-1}$ for $1\leq i\leq n$.
\end{lem}
\begin{proof}
Note that $GRS_k(\mathbf{a},\mathbf{v})$ has a generator matrix $G$  given by (\ref{matrix}).
Clearly, we have $G=G_1\Delta$, where
\begin{equation*}
G_1=\left(
    \begin{array}{ccccccc}
      1 & 1 & \cdots & 1 &     \\
      \alpha_1 &  \alpha_2 & \cdots &  \alpha_{n}     \\
      \alpha_1^2 &  \alpha_2^2 & \cdots &  \alpha_{n}^2     \\
      \vdots & \vdots & \ddots & \vdots&     \\
       \alpha_1^{k-1} &  \alpha_2^{k-1} & \cdots &  \alpha_{n}^{k-1}     \\
    \end{array}
  \right)
\end{equation*}
and $\Delta$ is the diagonal matrix
$$
\Delta=
\left(
           \begin{array}{cccc}
             v_1 & \hfill & \hfill & \hfill \\
            \hfill & v_2 & \hfill & \hfill \\
             \hfill & \hfill & \ddots & \hfill \\
             \hfill & \hfill & \hfill & v_n \\
           \end{array}
         \right).
$$
It follows that
$\mathbf{c}=\big(v_1f(\alpha_1), v_2f(\alpha_2), \cdots, v_nf(\alpha_n)\big)$ is contained in $GRS_k(\mathbf{a},\mathbf{v})^\perp$
if and only if
$$
G\mathbf{c}^T=\big(G_1\Delta\big){\mathbf{\mathbf{c}}}^T=G_1\big(\Delta{\mathbf{c}}^T\big)=G_1\big(v_1^2f(\alpha_1), v_2^2f(\alpha_2), \cdots, v_n^2f(\alpha_n)\big)^T=\mathbf{0},
$$
where $\mathbf{c}^T$ denotes the transpose of $\mathbf{c}$. Recall that
the dual of $GRS_{k}(\mathbf{a},\mathbf{1})$ is $GRS_{n-k}(\mathbf{a},\mathbf{u})$, where
$\mathbf{u}=(u_1,u_2,\cdots, u_n)$ with $u_i=\prod_{1\leq j\leq n,j\neq i}(\alpha_i-\alpha_j)^{-1}$ for $1\leq i\leq n$.
Now the desired result follows immediately from (\ref{important}).
\end{proof}
The arguments used in   Lemma \ref{lem} can be modified to obtain an analogous result for extended GRS codes.
\begin{lem}\label{lem-extended}
Suppose  $GRS_k(\mathbf{a},\mathbf{v},\infty)$ is the  extended  GRS code
associated with $\mathbf{a}=(\alpha_1,\alpha_2,\cdots,\alpha_q)$
and $\mathbf{v}=(v_1,v_2,\cdots,v_q)$,  as in (\ref{extended-grs}).
A typical codeword $\mathbf{c}=\big(v_1f(\alpha_1), \cdots, v_qf(\alpha_q),f_{k-1}\big)$ of $GRS_k(\mathbf{a},\mathbf{v},\infty)$
is contained in
$GRS_k(\mathbf{a},\mathbf{v},\infty)^\perp$ if and only if
a polynomial $g(X)\in\mathbb{F}_q[X]$ with  $\deg g(X)\leq q-k$ can be found such that
\begin{equation*}
\big(v_1^2f(\alpha_1),  \cdots, v_q^2f(\alpha_q),f_{k-1}\big)
=\big(g(\alpha_1),  \cdots, g(\alpha_q), g_{q-k}\big),
\end{equation*}
where $g_{q-k}$ is the coefficient of $X^{q-k}$ of the polynomial $g(X)$.
\end{lem}

We first construct LCD MDS codes of length $q+1$ by using Lemma \ref{lem-extended}.
\cite[Theorem IV.4]{Jin}
guarantees that
there exists a $k$-dimensional LCD
MDS code of length $q+1$ when $k$ is even.
The next result improves \cite[Theorem IV.4]{Jin} by removing the even $k$ constraint.
Recall that we always restrict ourself to $k$-dimensional  codes of length $n$ with
$1<k\leq \lfloor n/2\rfloor$.
\begin{Theorem}\label{generalized1}
Let $q>3$ be an odd  prime power. Then there exists a $k$-dimensional LCD
extended  GRS  code of length $q+1$ over $\mathbb{F}_q$.
\end{Theorem}
\begin{proof}
Label  the elements of  $\mathbb{F}_q=\{\alpha_1,\alpha_2, \cdots, \alpha_q\}$.
Let $\mathbf{a}=(\alpha_1,\alpha_2,\cdots,\alpha_q)$ and
let $\gamma$ be a primitive $(q-1)$th root of unity in $\mathbb{F}_q$, i.e., $q-1$ is the smallest positive integer
such that $\gamma^{q-1}=1$. Since $q>3$, we have $\gamma^2\neq1$.

We now consider two cases separately.

Case 1: $2\leq k<(q+1)/2$.
Let
$$\mathbf{v}=\big(v_1,v_2,\cdots,  v_{q-k+1},v_{q-k+2},\cdots,v_q\big),$$
where  $v_i=1$  for $1\leq i\leq q-k+1$ and
$v_i=\gamma$ for $q-k+2\leq i\leq q$.
Consider the extended GRS code $GRS_{k}(\mathbf{a},\mathbf{v},\infty)$ of length $q+1$ over $\mathbb{F}_q$ associated with $\mathbf{a}$
and $\mathbf{v}$, i.e.,
\begin{equation*}
\begin{split}
GRS_{k}\big(\mathbf{a},\mathbf{v},\infty\big)=\Big\{&\big(f(\alpha_1), f(\alpha_2), \cdots, f(\alpha_{q-k+1}), \\
& \gamma f(\alpha_{q-k+2}), \cdots, \gamma f(\alpha_q), f_{k-1}\big)\,\Big{|}\,f(X)\in \mathbb{F}_q[X], \deg f(X)\leq k-1\Big\},
\end{split}
\end{equation*}
where $f_{k-1}$ is the coefficient of $X^{k-1}$ of the polynomial $f(X)$.
We now aim to show that $GRS_{k}(\mathbf{a},\mathbf{v},\infty)$ is an LCD code, namely
$GRS_{k}(\mathbf{a},\mathbf{v},\infty)\bigcap GRS_{k}(\mathbf{a},\mathbf{v},\infty)^\perp=\{\mathbf{0}\}$.
For this purpose, let
$$\mathbf{c}=\big(f(\alpha_1), f(\alpha_2), \cdots, f(\alpha_{q-k+1}),  \gamma f(\alpha_{q-k+2}), \cdots, \gamma f(\alpha_q), f_{k-1}\big)$$
be an arbitrary element of $GRS_{k}(\mathbf{a},\mathbf{v},\infty)\bigcap GRS_{k}(\mathbf{a},\mathbf{v},\infty)^\perp$, where
$f(X)\in \mathbb{F}_q[X]$ with $\deg f(X)\leq k-1$.
Applying  Lemma \ref{lem-extended},
we see  that a polynomial $g(X)\in\mathbb{F}_q[X]$ with  $\deg g(X)\leq q-k$  can be found such that
\begin{equation}\label{extended}
\begin{split}
&\big(f(\alpha_1), \cdots, f(\alpha_{q-k+1}), \gamma^2 f(\alpha_{q-k+2}), \cdots, \gamma^2 f(\alpha_q), f_{k-1}\big)\\
&=\big(g(\alpha_1),  \cdots, g(\alpha_{q-k+1}),  g(\alpha_{q-k+2}), \cdots,g(\alpha_q), g_{q-k}\big),
\end{split}
\end{equation}
where $g_{q-k}$ is the coefficient of $X^{q-k}$ of the polynomial $g(X)$.
Our task is to show $\mathbf{c}=\mathbf{0}$, or equivalently  $f(X)=0$.
By $\deg f(X)\leq k-1\leq q-k$ and $\deg g(X)\leq q-k$,
(\ref{extended}) suggests
$f(X)=g(X)$, since $f(\alpha_i)=g(\alpha_i)$ for $1\leq i\leq q-k+1$.
However, the last coordinate of Equation (\ref{extended}) requires
$f_{k-1}=g_{q-k}$. This leads to $f_{k-1}=0$; otherwise we would have $k=(q+1)/2$, a contradiction.
Thus, $\deg f(X)\leq k-2$.
At the moment, Equation (\ref{extended}) also gives
$$
\gamma^2 f(\alpha_j)=g(\alpha_j)=f(\alpha_j)~~\hbox{for $q-k+2\leq j\leq q$},
$$
where the last equality holds because $f(X)=g(X)$.
It follows from $\gamma^2\neq1$ that $f(\alpha_j)=0$ for $q-k+2\leq j\leq q$.
This forces $f(X)=0$, because $f(X)\in \mathbb{F}_q[X]$ is a polynomial with $\deg f(X)\leq k-2$  and has $k-1$ distinct roots
$\alpha_{q-k+2}, \cdots, \alpha_q$.  We thus conclude   $\mathbf{c}=\mathbf{0}$, as desired.

Case 2: $k=(q+1)/2$. In this case,
let
$$\mathbf{v'}=\big(v_1,v_2,\cdots,  v_{k-1},v_{k},\cdots,v_q\big),$$
where  $v_i=1$  for $1\leq i\leq k-1$ and
$v_i=\gamma$ for $k\leq i\leq q$.
Consider the extended GRS code $GRS_{k}(\mathbf{a},\mathbf{v'},\infty)$ of length $q+1$ over $\mathbb{F}_q$ associated with $\mathbf{a}$
and $\mathbf{v'}$, i.e.,
\begin{equation*}
\begin{split}
GRS_{k}\big(\mathbf{a},\mathbf{v'},\infty\big)=\Big\{&\big(f(\alpha_1), f(\alpha_2), \cdots, f(\alpha_{k-1}), \\
& \gamma f(\alpha_{k}), \cdots, \gamma f(\alpha_q), f_{k-1}\big)\,\Big{|}\,f(X)\in \mathbb{F}_q[X], \deg f(X)\leq k-1\Big\},
\end{split}
\end{equation*}
where $f_{k-1}$ is the coefficient of $X^{k-1}$ of the polynomial $f(X)$.
Let
$$
\big(f(\alpha_1), f(\alpha_2), \cdots, f(\alpha_{k-1}),  \gamma f(\alpha_{k}), \cdots, \gamma f(\alpha_q), f_{k-1}\big)
$$
be an arbitrary element of $GRS_{k}\big(\mathbf{a},\mathbf{v'},\infty\big)$ with $\deg f(X)\leq k-1=(q-1)/2$.
Therefore, we have a polynomial $g(X)\in\mathbb{F}_q[X]$ with  $\deg g(X)\leq q-k=(q-1)/2$    such that
\begin{equation}\label{extendedd}
\big(f(\alpha_1),   \cdots, f(\alpha_{k-1}), \gamma^2 f(\alpha_{k}), \cdots, \gamma^2 f(\alpha_q), f_{k-1}\big)
=\big(g(\alpha_1),   \cdots, g(\alpha_{k-1}), g(\alpha_{k}), \cdots, g(\alpha_q), g_{q-k}\big),
\end{equation}
where $g_{q-k}$ is the coefficient of $X^{q-k}$ of the polynomial $g(X)$.
By Equation (\ref{extendedd}), one has $f_{k-1}=g_{q-k}$ which implies that $\deg\big(f(X)-g(X)\big)\leq k-2$.
Similar arguments as in Case 1 give  $f(X)=0$. We are done.
\end{proof}

Next we turn to construct LCD MDS codes from GRS codes of length $1<n\leq q$.
We obtain the following theorems.
\begin{Theorem}\label{divides-theorem}
Let $q>3$ be an odd  prime power.
If $n>1$ is a divisor of $q-1$, then
there exists a $k$-dimensional
LCD GRS code of length $n$ over $\mathbb{F}_q$.
\end{Theorem}
\begin{proof}
Since $n>1$ is a divisor of $q-1$, there exists a primitive $n$th root of unity $\omega$ in $\mathbb{F}_q$.
Take $\mathbf{a}=(\omega^0,\omega^1,\cdots,\omega^{n-1})$ and let
$\mathbf{v}=(v_1,\cdots,v_{n-k+1},v_{n-k+2}, \cdots, v_n)$, where $v_i=1$ for $1\leq i\leq n-k+1$ and
$v_i\notin\{-1,0,1\}$ for $n-k+2\leq i\leq n$.
Consider the GRS code $\mathcal{C}$ of length $n$ over $\mathbb{F}_q$ associated with $\mathbf{a}$
and $\mathbf{v}$ as follows
\begin{equation*}
\mathcal{C}=
\Big\{\big(f(\omega^0), \cdots, f(\omega^{n-k}), v_{n-k+2}f(\omega^{n-k+1}),
\cdots,  v_nf(\omega^{n-1})\big)\,\Big{|}\,f(X)\in \mathbb{F}_q[X], \deg f(X)\leq k-1\Big\}.
\end{equation*}
We claim that  $\mathcal{C}\bigcap \mathcal{C}^\perp=\{\mathbf{0}\}$. To see this,
let
$$\big(f(\omega^0), \cdots, f(\omega^{n-k}), v_{n-k+2}f(\omega^{n-k+1}),
\cdots,  v_nf(\omega^{n-1})\big)$$
be an arbitrary element of $\mathcal{C}\bigcap \mathcal{C}^\perp$.
Consequently,
$$
G_2\big(f(\omega^0), \cdots, f(\omega^{n-k}), v_{n-k+2}^2f(\omega^{n-k+1}),
\cdots,  v_n^2f(\omega^{n-1})\big)^T=\mathbf{0},
$$
where $G_2$ is the $k\times n$ matrix
$$
G_2=\left(
  \begin{array}{cccc}
    1 & 1 & \cdots & 1 \\
    1 & \omega & \cdots & \omega^{n-1} \\
    \vdots & \vdots & \ddots & \vdots \\
    1 & \omega^{k-1} & \cdots & \omega^{(k-1)(n-1)} \\
  \end{array}
\right).
$$
At this point, it is easy to see that
$$
\left(
  \begin{array}{cccc}
    1 & 1 & \cdots & 1 \\
    1 & \omega & \cdots & \omega^{n-1} \\
    \vdots & \vdots & \ddots & \vdots \\
    1 & \omega^{n-2} & \cdots & \omega^{(n-2)(n-1)} \\
  \end{array}
\right)
\left(
  \begin{array}{c}
    1 \\
    \omega \\
    \vdots \\
    \omega^{n-1} \\
  \end{array}
\right)=\mathbf{0}.
$$
This implies that
 a polynomial $g(X)\in\mathbb{F}_q[X]$ with  $\deg g(X)\leq n-k-1$ can be found such that
\begin{equation}\label{divides}
\big(f(\omega^0), \cdots, f(\omega^{n-k}), v_{n-k+2}^2f(\omega^{n-k+1}),
\cdots,  v_n^2f(\omega^{n-1})\big)
=\big(g(\omega^0), \omega g(\omega^1), \cdots, \omega^{n-1}g(\omega^{n-1})\big).
\end{equation}
The first $n-k+1$ coordinates of (\ref{divides}) give
$
f(\omega^i)=\omega^ig(\omega^i)
$
for $0\leq i\leq n-k$.
Since $k\leq\lfloor n/2\rfloor$, $\deg f(X)\leq k-1\leq n-k-1$ and $\deg g(X)\leq n-k-1$ , we have $f(X)=Xg(X)$.
In particular, $\deg g(X)\leq k-2$.
By the last $k-1$ coordinates of (\ref{divides}), we have that for any $n-k+2\leq j\leq n$,
$$
v_j^2f(\omega^{j-1})=v_j^2\omega^{j-1}g(\omega^{j-1})=\omega^{j-1}g(\omega^{j-1}).
$$
It follows from $v_j^2\neq1$ that  $g(\omega^{j-1})=0$. In other words, $g(X)$ has $k-1$ distinct roots, giving $g(X)=0$.
Thus, $f(X)=0$. The proof is complete.
\end{proof}

The following theorem
particularly  indicates that $k$-dimensional LCD MDS codes  of length $q$ over $\mathbb{F}_q$ exist for any $1<k<q$.

\begin{Theorem}\label{power}
Let $q=p^e>3$, where $p$ is an odd prime number and $e\geq1$ is an integer.
Then there exists a $k$-dimensional LCD
GRS  code of length $n=p^\ell$ over $\mathbb{F}_q$, where $\ell$ is a  positive integer with $1\leq \ell\leq e$.
\end{Theorem}
\begin{proof}
Let $H$ be an additive subgroup of $\mathbb{F}_{p^e}$ of order $n=p^\ell$, say $H=\{\alpha_1,\cdots, \alpha_n\}$.
Let $h$ be the product of all nonzero elements of $H$, namely
$$
h=\prod_{z\in H\setminus{\{0\}}}z.
$$
Take $\gamma\in \mathbb{F}_q^*$ with $\gamma^2\neq1$.
Let $\mathbf{a}=(\alpha_1,\alpha_2,\cdots,\alpha_n)$ and
let $$\mathbf{v}=\big(v_1, \cdots, v_{n-k}, v_{n-k+1}, \cdots,v_n\big),$$
where  $v_i=1$  for $1\leq i\leq n-k$ and $v_i=\gamma$ for $n-k+1\leq i\leq n$.
Consider the  GRS code $GRS_{k}(\mathbf{a},\mathbf{v})$ of length $n$ over $\mathbb{F}_q$ associated with $\mathbf{a}$
and $\mathbf{v}$, i.e.,
\begin{equation*}
GRS_{k}\big(\mathbf{a},\mathbf{v}\big)=\Big\{\big(f(\alpha_1), \cdots, f(\alpha_{n-k}),
\gamma f(\alpha_{n-k+1}),\cdots,\gamma f(\alpha_n)\big)\,\Big{|}\,f(X)\in \mathbb{F}_q[X], \deg f(X)\leq k-1\Big\}.
\end{equation*}
As in the proofs of the previous theorems,
we can show that $GRS_{k}(\mathbf{a},\mathbf{v})\bigcap GRS_{k}(\mathbf{a},\mathbf{v})^\perp=\{\mathbf{0}\}$.
Indeed, let $\mathbf{c}=\big(f(\alpha_1), \cdots, f(\alpha_{n-k}),
\gamma f(\alpha_{n-k+1}),\cdots,\gamma f(\alpha_n)\big)$ be an   element of
$GRS_{k}(\mathbf{a},\mathbf{v})\bigcap GRS_{k}(\mathbf{a},\mathbf{v})^\perp$.
Using Lemma \ref{lem}, we have that
a polynomial $g(X)\in\mathbb{F}_q[X]$ with  $\deg g(X)\leq n-k-1$ can be found such that
\begin{equation*}
\big(f(\alpha_1), \cdots,  f(\alpha_{n-k}), \gamma^2f(\alpha_{n-k+1}), \cdots, \gamma^2f(\alpha_n)\big)
=\big(u_1g(\alpha_1), u_2g(\alpha_2), \cdots, u_ng(\alpha_n)\big),
\end{equation*}
where $u_i=\prod_{1\leq j\leq n,j\neq i}(\alpha_i-\alpha_j)^{-1}$ for $1\leq i\leq n$.
However, for any $1\leq i\leq n$ we have
$$
u_i=\prod_{1\leq j\leq n,j\neq i}(\alpha_i-\alpha_j)^{-1}=h^{-1}.
$$
Therefore,
\begin{equation*}
\big(f(\alpha_1), \cdots,  f(\alpha_{n-k}), \gamma^2f(\alpha_{n-k+1}), \cdots, \gamma^2f(\alpha_n)\big)
=\big(h^{-1}g(\alpha_1), h^{-1}g(\alpha_2), \cdots, h^{-1}g(\alpha_n)\big).
\end{equation*}
The desired result can be obtained by applying arguments similar to those used in the proof of Theorem \ref{divides-theorem}.
We are done.
\end{proof}

Theorems \ref{divides-theorem} and \ref{power} tell us that
$k$-dimensional LCD MDS codes  of length $q-1$ and $q$ over $\mathbb{F}_q$ exist for any $k$, respectively.
The following result implies that if the code length $n$ is close to the alphabet size $q$,
then LCD GRS codes of length $n$ over $\mathbb{F}_q$ exist.
\begin{Theorem}\label{3.6}
Let $q>3$ be an odd  prime power and    let $n$ be a positive integer with $1<n<q$. If $1< k\leq \lfloor n/2\rfloor$ and $n+k\geq q+1$,
then  there exists a $k$-dimensional
LCD GRS code of length $n$ over $\mathbb{F}_q$.
\end{Theorem}
\begin{proof}
Let $\mathbf{a}=(\alpha_1,\alpha_2,\cdots,\alpha_n)$,
where $\alpha_i$ ($1\leq i\leq n$) are distinct elements of $\mathbb{F}_q$.
Label  the elements of  $\mathbb{F}_q=\{\alpha_1,\alpha_2, \cdots,\alpha_n, \alpha_{n+1},\cdots, \alpha_q\}$.
Recall that
the dual of $GRS_{k}(\mathbf{a},\mathbf{1})$ is $GRS_{n-k}(\mathbf{a},\mathbf{u})$, where
$\mathbf{u}=(u_1,u_2,\cdots, u_n)$ with $u_i=\prod_{1\leq j\leq n,j\neq i}(\alpha_i-\alpha_j)^{-1}$.
Let
$$\mathbf{v}=\big(v_1,v_2,\cdots, v_{q-k}, v_{q-k+1}, \cdots,v_n\big),$$
where  $v_i=1$  for $1\leq i\leq q-k$ and
$v_i$, $q-k+1\leq i\leq n$, is chosen such that  $-v_i^2\prod_{j=n+1}^q(\alpha_i-\alpha_{j})\neq u_i$.
Consider the GRS code $GRS_k(\mathbf{a},\mathbf{v})$ of length $n$ over $\mathbb{F}_q$ associated with $\mathbf{a}$
and $\mathbf{v}$, i.e.,
\begin{equation*}
GRS_k(\mathbf{a},\mathbf{v})=
\Big\{\big(f(\alpha_1), \cdots, f(\alpha_{q-k}), v_{q-k+1}f(\alpha_{q-k+1}),
\cdots,  v_nf(\alpha_n)\big)\,\Big{|}\,f(X)\in \mathbb{F}_q[X], \deg f(X)\leq k-1\Big\}.
\end{equation*}
We are left to show that $GRS_k(\mathbf{a},\mathbf{v})\bigcap GRS_k(\mathbf{a},\mathbf{v})^\perp=\{\mathbf{0}\}$. To this end,
let
$$\mathbf{c}=\big(f(\alpha_1), \cdots, f(\alpha_{q-k}), v_{q-k+1}f(\alpha_{q-k+1}),
\cdots,  v_nf(\alpha_n)\big)$$
be an arbitrary element of $GRS_k(\mathbf{a},\mathbf{v})\bigcap GRS_k(\mathbf{a},\mathbf{v})^\perp$.
It follows from Lemma \ref{lem}  that a polynomial $g(X)\in\mathbb{F}_q[X]$ with  $\deg g(X)\leq n-k-1$ can be found such that
\begin{equation}\label{k+n}
\big(f(\alpha_1), \cdots, f(\alpha_{q-k}), v_{q-k+1}^2f(\alpha_{q-k+1}),
\cdots,  v_n^2f(\alpha_n)\big)
=\big(u_1g(\alpha_1), u_2g(\alpha_2), \cdots, u_ng(\alpha_n)\big).
\end{equation}
As a consequence, the first $q-k$ coordinates of (\ref{k+n}) give
$
f(\alpha_i)=u_ig(\alpha_i)
$
for $1\leq i\leq q-k$; equivalently,
$$
f(\alpha_i)=\prod_{1\leq j\leq n,j\neq i}(\alpha_i-\alpha_j)^{-1}g(\alpha_i)=-\prod_{j=n+1}^q(\alpha_i-\alpha_j)g(\alpha_i),
$$
where the last equality holds because
$$\prod_{1\leq j\leq q,j\neq i}(\alpha_i-\alpha_j)=-1.$$
This suggests that the polynomial $f(X)+\prod_{j=n+1}^q(X-\alpha_{j})g(X)$ has $q-k$ distinct roots $\alpha_1,\cdots,\alpha_{q-k}$.
Note that
$$
\deg f(X)\leq k-1\leq n-k-1\leq q-k-1
$$
and
$$
\deg \prod_{j=n+1}^q(X-\alpha_{j})g(X)\leq (q-n)+(n-k-1)=q-k-1.
$$
These  imply
the degree of $f(X)+\prod_{j=n+1}^q(X-\alpha_{j})g(X)$ is at most $q-k-1$.
We thus conclude that
$$f(X)=-\prod_{j=n+1}^q(X-\alpha_{j})g(X).$$
In particular,
$
\deg f(X)=q-n+\deg g(X)\leq k-1,
$
giving $\deg g(X)\leq n+k-q-1$. However, the last $n+k-q$ coordinates of (\ref{k+n}) imply that for any $q-k+1\leq i\leq n$,
$$
v_i^2f(\alpha_i)=u_ig(\alpha_i)=-v_i^2\prod_{j=n+1}^q(\alpha_i-\alpha_{j})g(\alpha_i).
$$
We have $g(\alpha_i)=0$ since $v_{i}$ is chosen such that $-v_i^2\prod_{j=n+1}^q(\alpha_i-\alpha_{j})\neq u_i$.
Therefore $g(X)$ has $n+k-q$ distinct roots, which forces $g(X)=0$.
We finally conclude that $f(X)=0$, as desired.
\end{proof}

We conclude this section with the following theorem.
\begin{Theorem}\label{last-theorem}
Let $q>3$ be an odd  prime power and    let $n$ be a positive integer with $1<n<q$.
If $1< k\leq \lfloor n/2\rfloor$ and $2n-k<q\leq2n$,
then  there exists a $k$-dimensional
LCD GRS code of length $n$ over $\mathbb{F}_q$.
\end{Theorem}
\begin{proof}
Let $\mathbf{a}=(\alpha_1,\alpha_2,\cdots,\alpha_n)$,
where $\alpha_i$ ($1\leq i\leq n$) are distinct elements of $\mathbb{F}_q$.
Label  the elements of  $\mathbb{F}_q=\{\alpha_1,\alpha_2, \cdots,\alpha_n, \alpha_{n+1},\cdots, \alpha_q\}$.
The assumption $2n-k<q$ is equivalent to saying that $q-n>n-k$.
Let
$\mathbf{v}=\big(v_1, v_2,\cdots,v_n\big)$
with
$$
v_i=\prod_{j=1}^{n-k}(\alpha_i-\alpha_{n+j})
$$
for $1\leq i\leq n$.
Consider the GRS code $GRS_k(\mathbf{a},\mathbf{v})$ of length $n$ over $\mathbb{F}_q$ associated with $\mathbf{a}$
and $\mathbf{v}$, i.e.,
\begin{equation*}
GRS_k(\mathbf{a},\mathbf{v})=
\Big\{\big(v_1f(\alpha_1), v_2f(\alpha_2), \cdots, v_nf(\alpha_n)\big)\,\Big{|}\,f(X)\in \mathbb{F}_q[X], \deg f(X)\leq k-1\Big\}.
\end{equation*}
We claim that $GRS_k(\mathbf{a},\mathbf{v})\bigcap GRS_k(\mathbf{a},\mathbf{v})^\perp=\{\mathbf{0}\}$. To this end,
let
$$\mathbf{c}=\big(v_1f(\alpha_1), v_2f(\alpha_2), \cdots, v_nf(\alpha_n)\big)$$
be an arbitrary element of $GRS_k(\mathbf{a},\mathbf{v})\bigcap GRS_k(\mathbf{a},\mathbf{v})^\perp$.
It follows from Lemma \ref{lem}  that a polynomial $g(X)\in\mathbb{F}_q[X]$ with  $\deg g(X)\leq n-k-1$ can be found such that
\begin{equation*}
\big(v_1^2f(\alpha_1), v_2^2f(\alpha_2), \cdots, v_n^2f(\alpha_n)\big)
=\big(u_1g(\alpha_1), u_2g(\alpha_2), \cdots, u_ng(\alpha_n)\big).
\end{equation*}
This gives
$
v_i^2f(\alpha_i)=u_ig(\alpha_i)
$
for $1\leq i\leq n$.
By the definition of $v_i$,
$$
v_i^2f(\alpha_i)=\prod_{j=1}^{n-k}(\alpha_i-\alpha_{n+j})^2f(\alpha_i).
$$
On the other hand,
$$
u_ig(\alpha_i)=\prod_{1\leq j\leq n,j\neq i}(\alpha_i-\alpha_j)^{-1}g(\alpha_i)
=-\prod_{j=n+1}^q(\alpha_i-\alpha_j)g(\alpha_i).
$$
These lead to
\begin{equation}\label{last}
\prod_{j=1}^{n-k}(\alpha_i-\alpha_{n+j})f(\alpha_i)=-\prod_{j=n+(n-k+1)}^q(\alpha_i-\alpha_j)g(\alpha_i)
\end{equation}
for $1\leq i\leq n$.
At the moment we have
$$
\deg \Big(\big(\prod_{j=1}^{n-k}(X-\alpha_{n+j})\big)f(X)\Big)\leq(n-k)+(k-1)=n-1
$$
and
$$
\deg\Big(\big(\prod_{j=n+(n-k+1)}^q(X-\alpha_j)\big)g(X)\Big)\leq (q-n)-(n-k)+(n-k-1)=q-n-1\leq n-1,
$$
where the last inequality holds because $q\leq2n$ by our assumption.
Hence
the degree of
$$
\Big(\prod_{j=1}^{n-k}(X-\alpha_{n+j})\Big)f(X)+\Big(\prod_{j=n+(n-k+1)}^q(X-\alpha_j)\Big)g(X)
$$
is at most $n-1$.
We thus conclude from (\ref{last}) that
$$
\Big(\prod_{j=1}^{n-k}(X-\alpha_{n+j})\Big)f(X)=-\Big(\prod_{j=n+(n-k+1)}^q(X-\alpha_j)\Big)g(X).
$$
Obviously $\prod_{j=1}^{n-k}(X-\alpha_{n+j})$ is a divisor of $g(X)$. This forces $g(X)=0$, as $\deg g(X)\leq n-k-1$.
We are done.
\end{proof}

\noindent{\bf Acknowledgements}\quad
The research of Bocong Chen was supported by NSFC (Grant No.
11601158).
The research of Hongwei Liu was supported by NSFC (Grant No.
11171370) and self-determined research funds of CCNU from the colleges¡¯s basic research and operation of
MOE (GrantNo. CCNU14F01004).

\end{document}